\documentclass[11pt]{article}
\usepackage{mathrsfs}
\usepackage{amsmath}
\usepackage[colorlinks=true, allcolors=blue, urlcolor=black]{hyperref}

\makeatletter
\newcommand{\removelatexerror}{\let\@latex@error\@gobble}
\makeatother
\usepackage{amsmath}
\usepackage{amsthm}
\usepackage{amsfonts}
\usepackage{amssymb}
\usepackage{graphicx}
\usepackage{dsfont}
\usepackage{mathtools}
\usepackage{cleveref}
\usepackage{pgfplots}
\usepackage{bbm}
\usepackage[noend]{algpseudocode}
\usepackage{stmaryrd}

\usepackage{algorithm}
\usepackage{complexity}
\usepackage{enumitem}
\usepackage{tikz}
\usetikzlibrary{positioning}
\usepackage{thm-restate}
\usepackage{float}
\usetikzlibrary{calc}
\usepackage{thm-restate}
\usetikzlibrary{shapes,decorations.markings}
\usepackage{caption}
\usepackage{subcaption}
\usepackage{comment}

\usepackage{lipsum}
\usetikzlibrary{snakes}

\newtheorem{theorem}{Theorem}[section]

\newtheorem{definition}[theorem]{Definition}
\newtheorem{problem}[theorem]{Problem}
\newtheorem{claim}[theorem]{Claim}
\newtheorem{lemma}[theorem]{Lemma}

\newtheorem{corollary}[theorem]{Corollary}

\newcommand{\sk}{\mathsf{sk}}
\newcommand{\FS}{\mathsf{SupportFind}}

\usepackage[english]{babel}
\usepackage[T1]{fontenc}
\usepackage{caption}

\usepackage[margin=1in]{geometry}

\algblockdefx{Withprob}{EndWithprob}[1]{\textbf{with probability} #1 \textbf{do}}{}

\makeatletter
\ifthenelse{\equal{\ALG@noend}{t}}%
  {\algtext*{EndWithprob}}
  {}%
\makeatother

\title{A (Very) Nearly Optimal Sketch for \\ $k$-Edge Connectivity Certificates}
\author{ Pachara Sawettamalya\footnote{Department of Computer Science, Princeton University. Supported by NSF CAREER award CCF-233994. \ \href{mailto:ps3122@princeton.edu}{\url{ps3122@princeton.edu}}}  \and Huacheng Yu\footnote{Department of Computer Science, Princeton University. Supported by NSF CAREER award CCF-233994. \ \href{mailto:yuhch123@gmail.com}{\url{yuhch123@gmail.com}}}}
\date{}

\usepackage[backend=biber, doi=false, url=false, isbn=false, minalphanames = 3, maxalphanames = 3, maxbibnames=10]{biblatex}
\begin{document}
\maketitle



\begin{abstract}
In this note, we present a simple algorithm for computing a \emph{$k$-connectivity certificate} in dynamic graph streams. Our algorithm uses $O(n \log^2 n \cdot \max\{k, \log n \log k\})$ bits of space which improves upon the $O(kn \log^3 n)$-space algorithm of Ahn, Guha, and McGregor (SODA'12). For the values of $k$ that are truly sublinear, our space usage \emph{very nearly} matches the known lower bound  $\Omega(n \log^2 n \cdot \max\{k, \log n\})$ established by Nelson and Yu (SODA'19; implicit) and Robinson (DISC'24). In particular, our algorithm fully settles the space complexity at $\Theta(kn \log^2{n})$ for $k = \Omega(\log n \log \log n)$, and bridges the gap down to only a doubly-logarithmic factor of $O(\log \log n)$ for a smaller range of $k = o(\log n \log \log n)$.

\end{abstract}

\pagenumbering{roman}
\pagenumbering{arabic}

\section{Introduction}

The problem of computing an \emph{edge-connectivity certificate} concerns verifying whether a given graph satisfies a desired level of edge connectivity. Formally, given an $n$-vertex graph $G = (V, E)$ and a target edge-connectivity parameter $k \in \mathbb{N}$, the goal is to produce an object that certifies whether $G$ is $k$-edge-connected or not. In this work, we consider a notion of certificates that provide explicit evidence in either cases.

\begin{definition}[$k$-Connectivity Certificate]
Let $G = (V, E)$ be a graph on $n$ vertices and let $k$ be a positive integer. A \emph{$k$-connectivity certificate} is one of the following:
\begin{enumerate}
\item[$(+)$] If $G$ is $k$-edge-connected, the certificate is a spanning subgraph $H$ of $G$ such that $H$ itself is $k$-edge-connected. 
\item[$(-)$] If $G$ is not $k$-edge-connected, the certificate consists of a cut $(S, \overline{S})$ whose cut size is strictly less than $k$, along with the explicit list of edges crossing the cut denoted by $E(S, \overline{S})$.
\end{enumerate}
\end{definition}

We study the problem of computing a $k$-connectivity certificate in the dynamic graph stream \cite{FeigenbaumKMSZ05} where the input graph $G$ is presented as a sequence of edge insertions and deletions. The algorithm must process the stream in a single pass while using limited memory. At the end of the stream, it is required to output a valid $k$-connectivity certificate, as defined previously.

\begin{problem}[$k$-Connectivity Certificate in Dynamic Streams]
Given a connectivity parameter $k \in \mathbb{N}$ and an $n$-vertex graph $G$ presented as a stream of edge insertions and deletions, design a single-pass streaming algorithm that outputs a valid $k$-connectivity certificate of $G$.
\end{problem}

For the special case of $k = 1$, the problem indeed reduces to computing a spanning forest of the input graph $G$. The well-celebrated work of Ahn, Guha, and McGregor \cite{AhnGM12} shows that a spanning forest can be computed with high probability in the dynamic streaming model using only $O(n \log^3 n)$ space via an algorithm now known as the ``AGM sketch''.\footnote{Nelson and Yu \cite{NelsonY19} subsequently extends the AGM sketch to the regime of arbitrary error probability $\delta$ using space $O(n \log^2{n}\log(n/\delta))$.} For general values of $k$, this idea can be extended by invoking the AGM sketch recursively, using a total space of $O(kn \log^3 n)$.

On the other hand, a space lower bound for certifying positive instances (i.e., the $(+)$ case in the definition) can be derived via combining two prior results. First, Robinson \cite{Robinson23} establishes a lower bound of $\Omega(kn \log^2(n/k))$ bits, even when the algorithm is required to succeed with only constant probability. Separately, Nelson and Yu \cite{NelsonY19} show that even the simpler task of outputting a spanning forest requires $\Omega(n \log^3 n)$ space, thereby establishing the optimality of the AGM sketch for the spanning forest problem. While this lower bound does not immediately translate to the question of $k$-connectivity certificate—due to differences in problem formulation—their techniques can be adapted in a white-box manner to show that any dynamic streaming algorithm that certifies positive instances must also use $\Omega(n \log^3 n)$ space.\footnote{More precisely, to prove the $\Omega(n\log^3 n)$-space lower bound for spanning forest computation, Nelson and Yu construct a hard distribution over input graphs such that any algorithm computing a spanning forest must use at least $\Omega(n\log^3 n)$ bits of space. Upon closer inspection, a graph $G$ drawn from this hard distribution is, with overwhelming probability, $n^{\Omega(1)}$-connected, thereby satisfying the premise of our positive instances.} Together, these lower bounds leave an open gap of roughly $O(\min\{k, \log n\})$ compared to the $O(kn\log^3{n})$-space algorithm by \cite{AhnGM12}.

\paragraph{Our Contributions.} We present a simple algorithm for computing a $k$-connectivity certificate with improved space complexity in the dynamic streaming model.

\begin{theorem}
There is a single-pass dynamic streaming algorithm that computes a $k$-connectivity certificate using $O(n \log^2 n \cdot \max\{k, \log n \log k\})$ bits of space. The algorithm succeeds with probability at least $1 - n^{-5}$.
\label{main_thm}
\end{theorem}

This result very nearly closes the aforementioned gaps in space complexity. In particular, when the value of $k$ satisfies $\Omega(\log n \log \log n) \leq k \leq n^{1-\Omega(1)}$, our algorithm matches the lower bound of \cite{Robinson23} up to constant factors, settling the space complexity at $\Theta(kn \log^2{n})$. In the regime where $k = o(\log n \log \log n)$, we achieve an upper bound of $O(n \log^3{n} \log{k})$, leaving only a gap of $O(\log k) \leq O(\log \log n)$ compared to the lower bound by \cite{NelsonY19}.

We also highlight a by-product of our algorithm in the distributed sketching model. In this setting, the $n$-vertex graph $G$ is partitioned among $n$ players, each holding a single vertex along with its incident edges. All players simultaneously send messages to a central referee, who then computes a $k$-connectivity certificate of $G$. Our algorithm from \Cref{main_thm} extends naturally to this setting, as captured in the following corollary. For brevity, we omit the details of this extension.

\begin{corollary}
There is a randomized distributed sketching algorithm in which each player sends a message of length $O(\log^2 n \cdot \max\{k, \log n \log k\})$ to a central referee. With probability at least $1 - n^{-5}$, the referee can compute a $k$-connectivity certificate of the input graph.
\label{main_dist}
\end{corollary}

\subsection{Brief descriptions of our algorithm}

\paragraph{Previous Approaches.}  Nagamochi and Ibaraki \cite{NagamochiI92} observed that a $k$-connectivity certificate can be constructed from a graph $F = F_1 \cup \cdots \cup F_k$, where each $F_i$ is an arbitrary spanning forest of the remaining graph $G \setminus (F_1 \cup \cdots \cup F_{i-1})$. More precisely, it can be shown that $G$ is $k$-edge connected if and only if $F$ is $k$-edge-connected. The seminal work of \cite{AhnGM12} utilizes this construction in a clever way: by maintaining $k$ independent copies of the AGM sketch, one for recovering each $F_i$, we can obtain a $k$-connectivity certificate using $O(kn \log^3 n)$ space. 

A notable feature of the AGM sketch heavily exploited in this approach is its \emph{updatability}: after extracting the first $i-1$ spanning forests $F_1, \ldots, F_{i-1}$ from the first $i-1$ sketches, the ``fresh'' $i^{\text{th}}$ sketch can be updated by removing these forests, thereby querying it yields the next spanning forest $F_i$. Equivalently, this approach maintains $k$ independent sketches, each successively used to increase the edge-connectivity by one.

\paragraph{Our Approach.} Our algorithm takes a different avenue. Instead of maintaining $k$ subsketches as in~\cite{AhnGM12}, we use only $1 + \log_2 k$ subsketches, denoted $\mathcal{M}_0, \ldots,\mathcal{M}_{\log_2 k}$.  The first subsketch $\mathcal{M}_0$ is a standard AGM sketch that yields a spanning tree $H$, serving as a $1$-edge-connected subgraph of $G$.  Each subsequent subsketch $\mathcal{M}_i$ is then used for \emph{doubling} the edge connectivity of $H$ from $2^{i-1}$ to $2^i$, until it reaches $k$. It achieves that by identifying all small cuts of $H$ whose size is strictly less than $2^i$ and, for each such cut, ``fixing'' it by determining a set of $2^i$ edges in $G$ that cross the cut (if existed) and adding them to $H$.

Although our approach requires only a logarithmic number of sketches, each subsketch is required to perform an arguably much more demanding task: it needs to fix all small cuts by determining a set of edges crossing the cut. To this end, we utilize a known dynamic data structure called \emph{$k$-support-finds}: Given $k \geq \Omega(\log n)$ and a vector $x \in \mathbb{R}^m$ undergoing updates, a $k$-support-find sketch can recover up to $k$ nonzero coordinates of $x$ using $O(k \log^2 n)$ space.\footnote{See \Cref{sec:findsupport} for a formal definition.}  
Each $\mathcal{M}_i$ is implemented as a (variant of) $2^i$-support-find sketch so that fixing a small cut reduces to recovering $2^i$ supports of a corresponding vector.

Roughly speaking, the ``savings'' of our algorithm stem from two sources. First, to achieve $1/\poly(n)$ failure probability, a $k$-support-find sketch uses $O(k \log^2 n)$ space for large $k$, while a 1-support-find sketch is known to require $\Theta(\log^3 n)$ space \cite{KapralovNPWWY17}. In other words, for large $k$, by directly implementing a $k$-support find sketch instead of $k$ copies of a 1-support-find sketch, we have gained an $O(\log n)$ factor improvement.

Second, the subsketches are organized in a geometric hierarchy.  
Consequently, the total space usage does \emph{not} suffer the $O(\log{k})$ blowup and is dominated by the last sketch $\mathcal{M}_{\log_2 k}$. The doubling scheme also provides another key advantage: each subsketch $\mathcal{M}_i$ needs to fix only $\mathrm{poly}(n)$ small cuts (each being a 2-approximation of a minimum cut), rather than a trivial bound of $2^n$. This phenomenon gracefully allows a straightforward union bound over all the failure probabilities.

\subsection{Other related works}

Closely related to the edge-connectivity certificate problem is its \emph{value} version: determining the \emph{edge-connectivity} of a graph $G$, i.e., the maximum integer $k$ for which $G$ is $k$-edge-connected. This value coincides with the \emph{global minimum cut} of $G$, one of the most extensively studied problems in graph algorithms. Such a problem has been investigated across a wide range of computational models, including the sequential model~\cite{Karger93, KargerS96, Karger99, BenczurK15, KawarabayashiT19, Saranurak21, ChalermsookHNSS22, GoranciHNSTW23}, distributed models~\cite{NanongkaiS14, DagaHNS19, Parter19, Ghaffari020, Ghaffari0T20}, streaming and semi-streaming models~\cite{AhnGM12, MukhopadhyayN20, AssadiD21, GhoshS24}, and various query models~\cite{RubinsteinSW18, ApersEGLMN22, AnandSW25, kenneth-mordoch2025}.

It is also worth highlighting the \emph{vertex-connectivity} problem—a natural counterpart to edge-connectivity—which has recently received much traction ~\cite{NanongkaiSY19, LiNPSY21, SaranurakY22, JiangM23, BlikstadJMY25, JiangNSY25, JiangMY25}. In the dynamic streaming setting, improving upon \cite{GuhaMT15}, Assadi and Shah \cite{AssadiS23} gave a single-pass streaming algorithm for computing a $k$-vertex-connectivity certificate using $O(kn \log^4 n)$ space, along with a near-matching lower bound of $\Omega(kn)$. Narrowing this remaining polylogarithmic gap, analogous to the approach taken in our work, remains an intriguing open problem.

\section{Useful Subroutine: $\FS$}
\label{sec:findsupport}

We describe a subroutine that is used repeatedly in our algorithm, which we refer to under the umbrella term “support-finds.” This subroutine resembles the classical $\ell_0$-sampler and the universal relation problem (see, e.g., \cite{KarchmerRW95, FrahlingIS08, JowhariST11, KapralovNPWWY17}), but is adapted to a query-based setting. Specifically, rather than sampling from a single vector, the algorithm operates over a collection of vectors and responds to queries on subsets of this collection.

\begin{definition}[SupportFind]
Let $n,m,$ and $k$ be positive integers with $m \leq  \mathrm{poly}(n)$, and $\delta \in (0,1)$ be the error parameter. The dynamic streaming problem $\FS(k, n, m, \delta)$ is defined as follows:

Let $x_1, \dots, x_n \in \mathbb{Z}^m$ be vectors, each initialized to $0^m$ at the beginning of the stream. Each stream updates are of the form $(i, \Delta) \in [n] \times \mathbb{Z}^m$, where $\|\Delta\|_\infty \leq \mathrm{poly}(n)$ and the update modifies $x_i \leftarrow x_i + \Delta$. At the end of the stream, given a query set $S \subseteq [n]$, the algorithm must return $\min\{k, \|x_S\|_0\}$ distinct indices from the support of $x_S$, with probability at least $1 - \delta$, where $x_S := \sum_{i \in S} x_i$.
\end{definition}

It turns out that $\FS$ can be computed very efficiently in space.

\begin{lemma}
There exists a sketching algorithm that computes $\FS(k, n, m, \delta)$ using space $O(tn \log^2 m)$, where $t = O(\max\{k, \log(1/\delta)\})$. Specifically, there exists a sketch $\mathcal{M}$ consisting of $n$ sub-sketches $\sk(x_1), \dots, \sk(x_n)$, and a query algorithm $\mathcal{A}$ with the following properties:
\begin{enumerate}
    \item Each $\sk(x_i)$ is a sketch of vector $x_i$ that supports a sequence of updates to $x_i$. Moreover, each $\sk(x_i)$ uses $O(t \log^2 m)$ bits of space.
    \item Upon querying a set $S \subseteq [n]$, the algorithm $\mathcal{A}$ operates on $\sk(x_1), \dots, \sk(x_n)$ and $S$, and returns $\min\{k, \|x_S\|_0\}$ distinct indices in the support of $x_S$ with probability at least $1 - \delta$.
\end{enumerate}
\label{lem:find_support}
\end{lemma}

As a crucial remark, we observe that our sketch $\mathcal{M}$ gracefully supports \emph{non-adaptive} queries. That is, given $q$ non-adaptive sets $S_1, \dots, S_q$, the algorithm $\mathcal{A}$ answers each individual query correctly with probability at least $1 - \delta$. By a union bound, the probability that all $q$ queries are answered correctly simultaneously is at least $1 - q\delta$.

It is important to emphasize that we do \emph{not} claim any novelty regarding \Cref{lem:find_support}, as it closely follows the works of \cite{JowhariST11} and \cite{KapralovNPWWY17}. For completeness, we provide the proof below. In our actual algorithm for $k$-edge connectivity certificates, we shall apply the $\FS$ sketch in a black-box way. We encourage readers to skip ahead to \Cref{sec:alg} where we discuss the algorithm, and revisit the proof of \Cref{lem:find_support} later, if needed.

\subsection{Proof of \Cref{lem:find_support}}

We first discuss a subproblem called \emph{$\ell$-sparse recovery}. In this problem, we are given an additional parameter $\ell \in \mathbb{Z}^+$ and must output $x$ exactly whenever $\|x\|_0 \leq \ell$. This task can, in fact, be accomplished \emph{deterministically} using $O(\ell \log m)$ bits of space. More importantly, it suffices to maintain a sketch in the form of a vector $Ax$, where $A \in \mathbb{F}_p^{O(\ell) \times m}$ for a prime $p \leq \mathrm{poly}(m)$ is called a \emph{sketching matrix}. Specifically, we want the sketching matrix $A$ to be equipped with the property that $Ay \neq 0$ for any $y \ne 0^m$ with $\|y\|_0 \leq 2\ell$. This shall ensure that the mapping $x \mapsto Ax$ is injective over the domain of $\ell$-sparse vectors; hence, one can uniquely recover such $x$ from $Ax$. One concrete choice is to let $A$ be the transpose of a Vandermonde matrix over $\mathbb{F}_p$. 

We emphasize on the fact that such sketches are \emph{linear} as it is simply a mapping $x \mapsto Ax$. This means they support linear operations—e.g., the sketch of $x + y$ is simply the sum of the sketches of $x$ and $y$, etc. This linearity turns out to be the key defining feature of the upcoming algorithms.

Next, we restrict our attention to a special case of Lemma~\ref{lem:find_support} where we consider a single vector $x \in \mathbb{Z}^m$ which is initialized as $0^m$ and updated over time. At the end of the stream, we wish to recover $\min\{k, \|x\|_0\}$ distinct indices in the support of $x$ with probability at least $1 - \delta$. We claim that this problem can be solved using $O(t \log^2 m)$ bits of space, where $t = \max\{k, C \log(1/\delta)\}$ for a suitably large constant $C > 0$. It is also worth mentioning that \cite{KapralovNPWWY17} shows that the sketch size can be slightly improved to $O(t \log^2(m/t))$ bits if we are promised that $\|x\|_\infty \leq O(1)$. However, this subtle improvement turns out to be indifferent to the implementation of our main algorithm.

\paragraph{Sketch description.} 
For now, suppose we have a perfect probabilistic hash function $h : [m] \rightarrow [\log_2 m]$ such that $\Pr[h(i) = j] = 2^{-j}$ for each $i \in [m]$ and $j \in [\log_2 m]$. We also note that the total probability mass of $h$ sums to only $1-\frac{1}{m}$. To correct this, we allocate the remaining $\frac{1}{m}$ mass to the bottom-most level $j = \log_2 m$. Our sketch then proceeds as follows:
\begin{enumerate}
    \item Let $A \in \mathbb{F}_p^{O(t) \times m}$ be a sketching matrix with respect to $(32t)$-sparse recovery.
    \item For each level $j \in [\log_2 m]$, we maintain $Ax^{(j)}$, where $x^{(j)}$ denotes the restriction of $x$ to coordinates $i$ with $h(i) = j$ (i.e., all other coordinates are zeroed out).
\end{enumerate}

\paragraph{Updating the sketch.} 
For any update to index $i$ of the form $x_i \leftarrow x_i + u$, we compute $j = h(i)$ and update $Ax^{(j)} \leftarrow Ax^{(j)} + u \cdot Ae_i$, using the linearity of the sketch.

\paragraph{Space usage.} 
There are two primary sources of space complexity that we must account for. The first is the sketch size. There are $O(\log m)$ levels, and each level stores a sketch for $(32t)$-sparse recovery using $O(t \log m)$ bits. Hence, the total space used for sketches is $O(t \log^2 m)$.

The second source of space complexity arises from lifting the assumption of a “perfect” hash function $h$. Instead, we sample each $h(i)$ from an $O(t)$-wise independent hash family, which incurs an additional space cost of $O(t \log n)$ bits to store throughout the stream.\footnote{This suffices for our Chernoff bound analysis; see Theorem 5 in \cite{SchmidtSS95}.} However, we must also ensure that the resulting distribution over $h(i)$ is an $O(t)$-wise independent \emph{geometric} distribution. To simulate this, we proceed as follows: draw $x_1, \dots, x_m$ from an $O(t)$-wise uniform distribution over $[m]$, and for each $i$, assign $h(i) \leftarrow \lfloor \log_2(m / x_i) \rfloor + 1$.

\paragraph{Answering a query.} 
To recover $\min\{k, \|x\|_0\}$ distinct support indices of $x$, we proceed as follows:
\begin{enumerate}
    \item For $j = \log_2 m, \ldots, 1$:
    \begin{enumerate}
        \item Query $x^{(j)}$ with $(32t)$-sparse recovery. Let the answer be $a_j$.
        \item If $\|a_j\|_0 \geq t$, return $k$ arbitrary coordinates from the support of $a_j$ and terminate.
    \end{enumerate}
    \item If no level $j$ satisfies $\|a_j\|_0 \geq t$, compute $\hat{a} = \sum_{j\in[\log_2{m}]} a_j$ and return $\min\{k, \|\hat{a}\|_0\}$ supports of $\hat{a}$.
\end{enumerate}

\paragraph{Analysis.} 
We analyze correctness of algorithm by considering two cases:

\begin{enumerate}
    \item $\|x\|_0 \leq 2t$. In this case, for each level $j$, its support size is $\|x^{(j)}\|_0 \leq \|x\|_0 \leq 2t \leq 32t$, so each $x^{(j)}$ is recovered exactly via $(32t)$-sparse recovery. In other words, we have $a_j = x^{(j)}$ for every level $j$. If any level contains support size at least $t \geq k$, we output $k$ of them from Line 1a. Otherwise, we fully recover $\hat{a} = \sum_{j\in[\log_2{m}]} a_j = \sum_{j\in[\log_2{m}]} x^{(j)} = x$ and return any of its $\min\{k, \|x\|_0\}$ supports.

    \item $\|x\|_0 \geq 2t$. In this case, there exists a unique nonnegative integer $j^*$ such that $2t \leq 2^{-j^*} \|x\|_0 < 4t$. Let $E$ denote the event that $\|x^{(j)}\|_0 \geq 32t$ for some $j \geq j^*$, and let $F$ denote the event that $\|x^{(j^*)}\|_0 \leq t$. Observe that in the events that $E$ and $F$ did \emph{not} occur, the algorithm shall terminate at some iteration between $j = \log_2{m}$ and $j =j^*$ and produces the correct output. Thus, it suffices to bound the probability that neither $E$ nor $F$ occurs.

    For $F$, note that $\mathbb{E}\left(\|x^{(j^*)}\|_0\right) = 2^{-j^*}\|x\|_0  \geq 2t.$ By a Chernoff bound, we have $$\Pr(F) = \Pr\left( \|x^{(j^*)}\|_0 \leq t \right) \leq e^{-\Omega(t)}.$$

    For $E$, consider any $j \geq j^*$. We have $\mathbb{E}\left(\|x^{(j)}\|_0\right) = 2^{-j} \|x\|_0 \leq (4t) \cdot 2^{-(j - j^*)}$. Applying a Chernoff bound, we have $\Pr\left(\|x^{(j)}\|_0 \geq 32t\right) \leq \left(8 \cdot 2^{j - j^*}\right)^{-\Omega(t)}.$\footnote{We use the following version of the Chernoff bound: for $\lambda \geq 2e-1$, the upper tail probability is $\lambda^{-\Omega(\lambda \mu)}$ where $\mu$ is the expected sum of the independent Bernoulli random variables.} By a union bound over all $j \geq j^*$, we have
    $$ \Pr(E) \leq \sum_{j \geq j^*} \left(8 \cdot 2^{j - j^*}\right)^{-\Omega(t)} \leq e^{-\Omega(t)}.$$
    Combining both bounds, we get $\Pr(E \cup F) \leq e^{-\Omega(t)} \leq \delta,$ provided $t \geq C \log(1/\delta)$ for a sufficiently large constant $C > 0$.
\end{enumerate}

\paragraph{From a single vector to $n$ vectors.}
We now extend the construction to match the query version of Lemma~\ref{lem:find_support}. For each vector $x_i$, we define $\sk(x_i)$ to be the sketch as described above. The total space becomes $n \cdot O(t \log^2 m) = O(tn \log^2 m)$ as desired. To answer a query for a subset $S \subseteq [n]$, we use the linearity of the sketch: since $x_S = \sum_{i \in S} x_i$, we can compute $\sk(x_S) = \sum_{i \in S} \sk(x_i)$. The resulting sketch behaves identically to the single-vector case but for $x_S$, and thus we recover $\min\{k, \|x_S\|_0\}$ distinct support indices of $x_S$ with probability at least $1 - \delta$.

\section{Our Algorithm}
\label{sec:alg}

In this section, we describe our streaming algorithm for computing a $k$-connectivity certificate, thereby proving \Cref{main_thm}.

\subsection{Sketch description}

We adopt the framework of Ahn, Guha, and McGregor~\cite{AhnGM12} by representing a graph as a collection of signed vertex–edge incidence vectors. Specifically, let $G$ be a graph over a vertex set $V$ given as a stream of edge insertions and deletions, and let $\prec$ be a fixed total ordering over the $n$ vertices. For each vertex $v \in V$, we define a vector $x_v \in \mathbb{Z}^{\binom{n}{2}}$ to be its corresponding edge-incidence vector, where the $\binom{n}{2}$ coordinates are indexed by all \emph{ordered} vertex pairs $(a,b)$ with $a \prec b$. Formally, the vector $x_v$ is defined coordinate-wise as follows:

$$x_v(a,b) := \begin{cases}
+1 & \text{if } a = v \text{ and } (a,b) \in E, \\
-1 & \text{if } b = v \text{ and } (a,b) \in E, \\
0 & \text{otherwise.}
\end{cases}$$

The key observation of this representation is that the support of $x_S = \sum_{v \in S} x_v$ corresponds exactly to the set of edges crossing the cut $(S, \overline{S})$, denoted by $E(S,\overline{S})$. As a result, querying the supports of $x_S$ yields a set of edges that cross the cut $(S, \overline{S})$. The construction of our sketch will utilize this observation by including multiple independent instances of the $\FS$ sketch over a set of vectors $\{x_v\}_{v \in V}$. Formally, we describe our sketches below. Note that we assume for simplicity that $k = 2^R$ is a power of two, though this assumption can be easily lifted.
\begin{table}[H]
    \centering
    \begin{tabular}{|p{15.5cm}|}
    \hline ~\\
    \multicolumn{1}{|c|}{\textbf{A sketch for $k$-connectivity certificate}} \\
    \vspace{1.5mm}
    We maintain a set of independent sketches $\mathcal{M}_0, \mathcal{M}_1,...,\mathcal{M}_R$ where:
    \begin{itemize}
        \item $\mathcal{M}_0$ is a black-box AGM sketch that computes a spanning forest with error $16n^{-6}$.
        \item For each $r \in [R]$, the sketch $\mathcal{M}_r$ is a $\FS(2^r, n, \binom{n}{2}, n^{-10})$ sketch over a set of vectors $x_1,...,x_n \in \mathbb{Z}^{\binom{n}{2}}$ which errs with probability $n^{-10}$ for each (non-adaptive) query. Notably, each sketch $\mathcal{M}_r$ consists of $n$ sub-sketches (as in \Cref{lem:find_support}) denoted by $\sk^{(r)}(v)$ for each $v \in V$.
    \end{itemize}
    \vspace{1mm}
    \\
    \hline
    \end{tabular}
    \caption{The sketch (i.e. memory contents) of our algorithm for $k$-connectivity certificate.}
    \label{table:elas-basic-ops}
\end{table}

\paragraph{Processing edge updates.} For any insertion of an edge $e = (u, v)$ with $u \prec v$, we update the sketches as follows. First, we update the AGM sketch $\mathcal{M}_0$ with the insertion of $e$. Then, for each subsequent sketch $\mathcal{M}_r$, we update the corresponding sub-sketches: we update coordinate $(u,v)$ in $\sk^{(r)}(u)$ with ``$+1$'', and update coordinate $(u,v)$ in $\sk^{(r)}(v)$ with ``$-1$''.

Similarly, for a deletion of the edge $e = (u,v)$, we perform the reverse operations. First, we update the AGM sketch $\mathcal{M}_0$ with the deletion of $e$. Then, for each subsequent sketch $\mathcal{M}_r$, we update the corresponding sub-sketches: we update coordinate $(u,v)$ in $\sk^{(r)}(u)$ with ``$-1$'', and update coordinate $(u,v)$ in $\sk^{(r)}(v)$ with ``$+1$''

The following claim argues the space complexity of our sketch.

\begin{claim} The sketch size is $O(n \log^2{n}\cdot \max\{k, \log n \log k\})$.
\end{claim}
\begin{proof} The first sketch $\mathcal{M}_0$ uses space $O(n \log^3 n)$ via the guarantee of the AGM sketch. For each $ r \in [R]$, following \Cref{lem:find_support}, the sketch $\mathcal{M}_r$ uses space $O(t_r n \log^2 n)$ where $t_r = \max\{2^r, 10 \log n\}$. Therefore, the total space is $O(n \log^2 n) \cdot \sum_{r \in [R]} t_r$. Now consider:
\begin{align*}
    \sum_{r \in [R]} t_r & = \sum_{r \in [R]} \max\{2^r, 10 \log n\} = \begin{cases} O(R \log n) & \text{; if $2^R \leq 10 \log n$} \\
    O(2^R + \log n \log \log n) & \text{; if $2^R > 10 \log n$}.\end{cases}
\end{align*}
With $2^R = k$, this calculation simplifies to $\sum_{r \in [R]} t_r \leq O(\max\{k, \log n \log k\})$. Hence, our sketch size becomes $O(n \log^2{n}\cdot \max\{k, \log n \log k\})$, as wished.
\end{proof}

\subsection{Querying a certificate}

Finally, we describe how the streaming algorithm constructs a $k$-connectivity certificate from its memory contents $\mathcal{M}_0, \mathcal{M}_1, \ldots, \mathcal{M}_R$. The algorithm proceeds inductively across $R = \log_2 k$ iterations. In the first iteration, we use $\mathcal{M}_0$ (i.e., the AGM sketch) to construct a spanning forest of $G$, which serves as a $1$-connectivity certificate. If the resulting forest is not connected, we have immediately found a certificate for a negative instance. Otherwise, we proceed to the next phase.

For each subsequent iteration $r \in [R]$, the algorithm attempts to augment the $2^{r-1}$-connectivity certificate obtained in the previous round to a $2^r$-connectivity certificate using the sketch $\mathcal{M}_r$. Specifically, we identify all cuts in the $2^{r-1}$-connectivity certificate (a subgraph of $G$) whose cut size is strictly less than $2^r$. For each such cut, we query $\mathcal{M}_r$ to retrieve a set of $2^r$ edges crossing it. Note that this set may include some of the edges we already found in the previous iterations. If \emph{strictly} fewer than $2^r \leq k$ such edges exist for any cut, we have discovered a small cut of size strictly less than $k$, thereby certifying that the input graph is not $k$-connected. Otherwise, we augment the subgraph by adding all such edges, ensuring that no cut of size less than $2^r$ remains.

This process continues for all $r = 1$ to $R$. If the algorithm does not detect any small cuts throughout the iterations, the resulting subgraph upon termination must be $k$-edge connected, thereby certifying a positive instance.



\renewcommand{\figurename}{Algorithm}

\begin{figure}[H]
    \centering
    \begin{tabular}{|p{16cm}|}
    \hline ~\\
   \multicolumn{1}{|c|}
     {} \\
    \vspace{-7.5mm} 
    \textbf{Algorithm 1:} Query a $k$-connectivity certificate. \\

    \textbf{Memory contents:} An AGM sketch $\mathcal{M}_0$, and support-find sketches $\mathcal{M}_1,\ldots, \mathcal{M}_r$. \\
    \textbf{Output:} A $k$-connectivity certificate of $G$. \\~\\

    \textbf{Procedures:}
    \begin{enumerate}
        \item query $\mathcal{M}_0$ (i.e. the AGM sketch) to obtain a spanning forest $G_0$.
        \item \textbf{if} $G_0$ is not connected \textbf{do}
        \item \hspace{0.5cm} let $C \subset V$ be a connected component of $G_0$. 
        \item \hspace{0.5cm} output the cut $(C,\overline{C})$ that has no crossing edges. Terminate immediately.
        \item \textbf{for} $r=1,\ldots, R$ \textbf{do}
        \item \hspace{0.5cm} list all $S \subseteq V$ whose cut size in $G_{r-1}$ is less than $2^r$. Let $\mathcal{S}_r$ be a set of those cuts.
        \item \hspace{0.5cm} \textbf{for} {$S \in \mathcal{S}_r$} \textbf{do}
        \item \hspace{1cm} query $\min\{2^r,|E(S, \overline{S})|\} $ edges in $E(S, \overline{S})$ via sketch $\mathcal{M}_r$. Let those edges be $E_S$.
        \item \hspace{1cm} \textbf{if} $|E_S| < 2^r$ \textbf{do}
        \item \hspace{1.5cm} output the cut $(S,\overline{S})$ and a set of crossing edges $E_S$. Terminate immediately.
        \item \hspace{0.5cm} update $G_{r} \leftarrow G_{r-1} \cup \left(\bigcup_{S \in \mathcal{S}_r} E_S\right)$
        \item output $G_R$
    \end{enumerate}
    \vspace{0mm}
    \\
    \hline
    \end{tabular}
    \caption{A query algorithm for computing a $k$-connectivity certificate.}
    \label{alg:k-conn-sketch}
\end{figure}

It only remains to justify the correctness and success probability of our algorithm. This is done via a combination of \Cref{clm:correctness_positive} (for positive instances) and \Cref{clm:correctness_negative} (for negative instances). In what follows, we refer to ``iteration $r$'' as in Line 6-11 of Algorithm \ref{alg:k-conn-sketch} with respect to the value of $r$.


\begin{claim} If $G$ is $k$-edge connected, then with probability $1-n^{-5}$, Algorithm \ref{alg:k-conn-sketch} outputs $G_R$ via Line 12 that is $k$-edge connected.
\label{clm:correctness_positive}
\end{claim}

\begin{proof} Define $\mathcal{E}_0$ to be the event that $G_0$ is connected, which occurs with probability at least $1-16n^{-6}$ via the guarantees of the AGM sketch. For each $r \in [R]$, let $\mathcal{E}_r$ be an event that the (1) the algorithm enters iteration $r$, (2) the algorithm completes iteration $r$, i.e. it does not terminate at Line 10, and (3) towards the end of the iteration, $G_r$ is $2^r$-connected. We will show that $\Pr(\mathcal{E}_{r} \mid \mathcal{E}_{\leq r-1}) \geq 1-16n^{-6}$ for any $r \in [R]$.

Consider any iteration $r \geq 1$. Conditioning on $\mathcal{E}_{r-1}$ (and $\mathcal{E}_{<r-1}$ although unnecessary) implies that upon finishing iteration $r-1$, the graph $G_{r-1}$ is $2^{r-1}$-connected. Now consider iteration $r$. For each $S \in \mathcal{S}_r$, we can query $E_S := \min\{2^r, |E(S,\overline{S})|\}$ edges among $E_G(S, \overline{S})$ via querying $x_S$ to the sketch $\mathcal{M}_r$. This succeeds with probability at least $1-n^{-10}$. Furthermore, since $G$ is $k$-edge connected, we have $|E(S,\overline{S})| \geq k \geq 2^r$. Hence, when it is the case that such query is correct, we always have $|E_S| = 2^r$; hence, the algorithm does not terminate at Line 10.

Moreover, observe that the queries made in iteration $r$ to the sketch $\mathcal{M}_r$ is non-adaptive. This allows us to perform the ``union bound'' over all queries. Furthermore, we claim that the number of such non-adaptive queries is $|\mathcal{S}_r| \leq 16n^4$.
This is because the min-cut size of $G_{r-1}$ is at least $2^{r-1}$, hence, a cut of size at most $2^r$ is a $2$-approximate min-cut.
\cite{KargerS96} shows that for any $\alpha \geq 1$, the number of $\alpha$-approximate min-cuts is bounded by $(2n)^{2\alpha}$.
Via union bounds, $\mathcal{M}_r$ answers all queries correctly with probability at least $1-16n^{-6}$. When that is the case, $G_r = G_{r-1} \cup \left(\bigcup_{S \in \mathcal{S}_r} \mathcal{E}_S\right)$ has no cuts of size smaller than $2^r$ remained; thus, $G_r$ is $2^r$-connected. This proves that $\Pr\left(\mathcal{E}_{r} \mid \mathcal{E}_{\leq r-1}\right) \geq 1-16n^{-6}$.

Finally, notice that if $\mathcal{E}_0 \land \ldots \land \mathcal{E}_R$, the algorithm shall eventually output $G_R$ via Line 12 that is $k$-edge connected. This occurs with probability 
\begin{align*}
    \Pr(\mathcal{E}_0 \land \ldots \land \mathcal{E}_R) & = \prod_{r = 0}^R \Pr(\mathcal{E}_{r} \mid \mathcal{E}_{\leq r-1}) \geq (1-16n^{-6})^{R+1} \geq 1-n^{-5},
\end{align*}
as wished.
\end{proof}

For $G$ that is \emph{not} $k$-connected, we have the following claim.

\begin{claim} With probability $1-n^{-5}$, one of the following occurs.
\begin{enumerate}
    \item If $G$ is not connected, then Algorithm \ref{alg:k-conn-sketch} outputs a cut $(C,\overline{C})$ that has no crossing edges via Line 4.
    \item If $G$ is connected but is not $k$-edge connected, then at some point Algorithm \ref{alg:k-conn-sketch} outputs a cut $(S,\overline{S})$ via Line 10 whose cut size is less than $k$, along with the list of the crossing edges.
\end{enumerate}
\label{clm:correctness_negative}
\end{claim}

\begin{proof} We condition on the event that the AGM sketch $\mathcal{M}_0$ produces a spanning forest $G_0$ correctly which happens with probability $1-16n^{-6}$. If $G$ is not connected, the cut $(C,\overline{C})$ via Line 4 does not have any crossing edges; thus, is a correct certificate. We then terminate immediately.

Otherwise, $G$ is connected. In this case, we proceed to Line 5. Let $\lambda \in  [k-1]$ be the value of $G$'s minimum cut which is attained by a min-cut $(S', \overline{S'})$. Denote $r \in [R]$ be such that $2^{r-1} \leq \lambda < 2^{r}$. We can show that all queries in rounds $\leq r-1$ are answered correctly with probability at least $(1-16n^{-6})^{r}$. We further condition of that occurs. As a result, upon the completion of iteration $r-1$, the algorithm has not terminated, and $G_{r-1}$ is $2^{r-1}$-edge connected.\footnote{The proofs are nearly identical to the proof of \Cref{clm:correctness_positive}; thus we shall omit due to succinctness.}

Consider the next iteration $r$. Since $|\mathcal{S}_r| \leq 16n^4$, a union bound implies that every query in this iteration is also answered correctly with probability at least $1 - 16n^{-6}$. We further condition on this event. To this point, all the conditioning events occur with probability at least $(1 - 16n^{-6})^{r + 1} \geq 1 - n^{-5}$. For the remainder of this proof, we will show that when that is the case, the algorithm correctly outputs a negative certificate.

We first show that $S' \in \mathcal{S}_r$. To see this, observe that conditioned on all prior queries in iterations $\leq r - 1$ being answered correctly, $G_{r - 1}$ is a subgraph of $G$. Hence, the cut size of $(S', \overline{S'})$ in $G_{r - 1}$ is at most $|E(S', \overline{S'})| = \lambda < 2^r$, implying that $S' \in \mathcal{S}_r$.

Next, we claim that the algorithm must terminate during iteration $r$, while processing some $S \in \mathcal{S}_r$. If it terminates before reaching $S = S'$, we are done. Otherwise, consider the point where $S = S'$. Since the query correctly returns $\min\{2^r, |E(S', \overline{S'})|\} = \min\{2^r, \lambda\} = \lambda$ edges from $E(S', \overline{S'})$, we must have $E_{S'} = E(S', \overline{S'})$ of size $ \lambda < 2^r$. Therefore, the algorithm must terminate.

Finally, we claim that whenever the algorithm terminates in iteration $r$ with respect to some $S \in \mathcal{S}_r$ (which may or may not be the same as $S'$), it correctly outputs a negative certificate at Line 10. Again, we rely on the fact that the query correctly returns $E_S :=$ the set of $\min\{2^r, |E(S, \overline{S})|\}$ edges from $E(S, \overline{S})$. The fact that the algorithm terminates implies that $|E_S| < 2^r \leq k$ which also implies $E_S = E(S, \overline{S})$. Thus, the algorithm shall output a cut $(S, \overline{S})$ of size less than $k$ along with the set of crossing edges $E_S = E(S, \overline{S})$. This constitutes a valid negative certificate.
\end{proof}

\section{Open Problems}

The most immediate open question is closing the gap for computing $k$-connectivity certificates in dynamic streams, which exists only when $k = o(\log n \log \log n)$. In this regime, our \Cref{main_thm} gives an upper bound of $O(n \log^3 n \log k)$, while a lower bound of $\Omega(n \log^3 n)$ follows from \cite{NelsonY19}, leaving a gap of just $O(\log k) = O(\log \log n)$. Notably, the same gap appears in the distributed sketching model.

We can also consider the decision variant which asks whether the graph is $k$-connected. Using simple reductions from \Cref{main_thm} and \Cref{main_dist}, we obtain algorithms with matching upper bounds. However, the only known lower bound is in the deterministic distributed sketching model, established by a very recent work of Robinson and Tan \cite{RobinsonT25} who prove an $\Omega(k)$ lower bound on the worst-case message length. Establishing lower bounds in other settings—distributed or streaming, deterministic or randomized—remains an intriguing open question.

A special case of such a problem arises when $k = 1$, reducing the task to deciding whether the graph is connected. The AGM sketch by Ahn, Guha, and McGregor \cite{AhnGM12} implies an $O(n \log^3 n)$-space algorithm for this problem in the dynamic streaming model. However, whether this space bound is optimal remains an open question. The best known lower bound to date is $\Omega(n \log n)$, due to Sun and Woodruff \cite{SunW15}, which also holds for insertion-only streams.

Finally, the approach of \cite{AhnGM12} naturally extends to the distributed sketching model, achieving a worst-case message length of $O(\log^3 n)$. More recently, Yu \cite{Yu21} proved that this bound is tight in terms of \emph{average} message length, thereby settling the question in the randomized setting. However, no nontrivial \emph{deterministic} algorithm is currently known—beyond the naive strategy of having each vertex send its entire neighborhood to the referee, nor any non-trivial lower bounds.

\section{Acknowledgement}

We thank the anonymous SOSA 2026 reviewers for their valuable feedback, most of which has been incorporated into this draft.

\printbibliography


\end{document}